\documentclass[conference]{IEEEtran}
\ifCLASSINFOpdf
\else
\fi
\usepackage{amsthm}
\usepackage{makeidx}
\usepackage{examples}
\usepackage[pdftex]{graphicx}  
\usepackage{amssymb}
\usepackage{graphicx}
\usepackage{algorithm}
\usepackage{epstopdf}
\usepackage{color}
\usepackage{url}
\font\msbm=msbm10 at 10pt
\newcommand{\ZZ}{\mbox{\msbm Z}}

\newcommand{\CC}{\mbox{\msbm C}}
\newcommand{\FF}{\mbox{\msbm F}}
\def \Z {{\ZZ}}
\def \C {{\CC}}

\def \F {{\FF}}
\def \x {{\bf x}}
\def \y {{\bf y}}

\begin{document}
\title{On Weak Dress Codes for Cloud Storage}

\author{\IEEEauthorblockN{Manish K. Gupta}
\IEEEauthorblockA{Laboratory of Natural Information Processing\\
Dhirubhai Ambani Institute of Information\\ and Communication Technology\\
Gandhinagar, Gujarat, 382007 India\\
Email: mankg@computer.org}
\and
\IEEEauthorblockN{Anupam Agrawal}
\IEEEauthorblockA{Infosys Ltd.\\
Pune, India\\
Email: anupam$\_$agarwal02@infosys.com}
\and
\IEEEauthorblockN{Deepak Yadav}
\IEEEauthorblockA{Sapient-Nitro\\
Gurgaon, Delhi, India\\
Email:dyadav4@sapient.com
}}


%Telephone: (800) 555--1212, Fax: (888) 555--1212}


\maketitle

\begin{abstract}
%\boldmath
In a distributed storage network, reliability and bandwidth optimization can be provided by regenerating codes. Recently  table based regenerating codes viz. DRESS (Distributed Replication-based Exact Simple Storage) codes has been proposed which also optimizes the disk I/O. Dress codes consists of an outer MDS code with an inner fractional repetition (FR) code with replication degree $\rho$.  Several constructions of FR codes based on regular graphs, resolvable designs and bipartite graphs are known. This paper presents a simple modular construction of FR codes. We also generalize the concept of FR codes to weak fractional repetition (WFR) codes where each node has different number of packets. We present a construction of  WFR codes based on partial regular graph. Finally we present a simple generalized ring construction of both strong and weak fractional repetition codes. 
\end{abstract}

% \fi
\IEEEpeerreviewmaketitle
\section{Introduction}
\IEEEPARstart{C}{loud} computing has emerged as a fascinating area of computing  providing various services such as on demand videos, music, computing etc. to end users from any computing device.  Usually data of the user is stored on different data centers in a distributed fashion.  Many such commercial cloud services are available from Amazon, Google and Apple etc. \cite{amazonec2,icloud,skydrive}. 
%Some of  the application companies such as Dropbox rents storage from the bigger cloud data storage companies such as Amazon. 
To provide seamless end user experience  for different applications such as Gmail etc. one has to optimize  various parameters of the storage system such as reliability, scalability, storage, bandwidth, disk I/O, etc.  In such a $distributed ~storage ~system$ (DSS) data is stored on a set of nodes in cloud network, which are not reliable independently. Redundancy is added in network to increase the reliability of the system. Google File System (GFS) is one such system which uses replication of data into various chunk servers (large data chunk files) together with replication of logs and master index for reliability~\cite{Ghemawat03}.  Although this scheme is simple but inefficient in terms of utilizaion of various storage system  parameters like bandwidth. One can provide reliability to such systems  using erasure coding which uses an MDS (Maximum Distance Separable) code ~\cite{patterson} however these codes fails to optimize bandwidth, complexity and disk I/O.  To optimize these parameters  recently regenerating codes were introduced and studied in detail by many researchers ~\cite{dgwr7,survey,dress11,DBLP:journals/corr/abs-1211-1932,2013arXiv1302.0744K,XorbasVLDB}.
%Data can be stored in a single node or disk but they are not able to provide the data rates required for application that process large number of random access. Distributing the data on different nodes in a distributed system provides data availability and avoids congestion in the network. 

Consider a distributed storage system (DSS) consisting of $n$ nodes, where each node can store $\alpha$ symbols  from $\F_q$ (finite field with $q$ number of elements). A file size of $B$ symbols can be stored in such a DSS either using simply a replication (as it is done currently in many cloud storage system) of packets or by using  an MDS code such that the file  can be retrieved using erasure codes by contacting any $k$ nodes out of those $n$ nodes in network ($k < n$) ~\cite{rr10}.  When a node fails, a new code can be generated (known as \it repair \rm of a failed node).  In a simple replication system one can download it from any other  node having a copy of the data. This is good bandwidth wise and disk I/O wise but it can tolerates only few errors as compare to an erasure code system (where a replacement node can contact to $k$ nodes and extract the data of that failed node) which can tolerates more failures of nodes but it is inefficient for bandwidth and disk I/O.  To minimize the bandwidth, the concept of regenerating code was introduced by Dimakis et al. \cite{dgwr7}. In regenerating codes instead of downloading whole file $B$ (by contacting all $k$ nodes), a node can be repaired by contacting $d$ nodes ($d$ is known as the repair degree) and downloading $\beta$ packets from each of them out of $\alpha$ packets which are stored on each node.  This new generated data is functionally equivalent to lost data. Thus repair bandwidth becomes $d \beta$.  Regenerating codes over $\F_q$ are characterized by parameters: $\{ \F_q, [n, k, d], [\alpha,\beta, B] \}$.  Bandwidth can be reduced further if each node is willing to store more \cite{survey}. Thus there is a tradeoff between storage ($\alpha$) and repair bandwidth ($d \beta$).  For internet  applications one can search for minimum bandwidth repair (MBR) regenerating codes and for archival applications one can search for minimum storage repair (MSR) codes. These are two extremal points on the storage-repair bandwidth tradeoff curve \cite{5961826,survey}.  For regenerating codes it is well known that
\begin{equation}  
k \leq  d \leq (n-1). 
\label{eq:1}
\end{equation}
Cut-set bound of network coding gives the following ~\cite{shah2012distributed} relation between the parameters of regenerating codes. 
\begin{equation}
B \leq \sum_{i=0}^ {k-1}  \min \{  \alpha ,( d -  i) \beta \}. 
\label{eq:2}
\end{equation}
Clearly we want to minimize both $\alpha$ as well as $\beta$ since minimizing $\alpha$ leads to minimum storage solution and minimizing $\beta$ leads to a solution which minimizes repair bandwidth.  There is a trade off between choices of $\alpha$ and $\beta$ as shown in \cite{shah2012distributed}. For regenerating codes to be optimal, equality holds in (\ref{eq:2})
and for the MBR (Minimum Bandwidth Regeneration) point, the parameter of the code must satisfy (\ref{eq:4}) (see  \cite{shah2012distributed}).
\begin{equation} 
B = \left(  kd - { k(k-1) \choose 2} \right)  \beta .    
\label {eq:4}
\end{equation}
For $ \beta = 1 $, the storage capacity (the limit on the maximum file size that can be delivered to any user contacting any $k$ out of $n$ nodes) in a DSS is given by \cite{rnvr9a}
\begin{equation}  C_{MBR}(n,k,d) = kd- {k \choose 2} \label {eq:5}. \end {equation} 
Constructions of regenerating codes for $\beta =1$ achieving storage capacity  (\ref{eq:5}) are of special interest. The problem of regeneration becomes harder if one wants to have an exact repair of the failed node. 
%Review on literature on different constructions
Rashmi et al. presented an explicit construction of exact-Minimum Bandwidth Repair (MBR) codes for the parameter set $[n, k = n-2, d = n-1]$ by fully connected graph \cite{rnvr9a}. An explicit MBR code is also given for all feasible values of the system parameters $[ n, k, d ]$ and MSR codes for all $2k-2 \leq d \leq n-1$ based on a common product matrix framework \cite{rnvr9a}.  
%Cadmabe et al ....
%%%% Introduction to our problem
El Rouayheb et al. \cite{rr10} has shown the construction of optimal regenerating codes (called Fractional Repetition Codes) by using table-based repair model, which is basically a generalization of the codes given in \cite{rnvr9a}. El Rouayheb et al. constructed the regenerating codes for single node failure based on regular graph and multiple node failure based on Steiner systems. These codes are fast and low-complexity repair codes. It simply reads only one of its stored packet and forward it to the replacement node with no additional processing. This property is referred as \it uncoded repair. \rm  By combining Fractional Repetition code with an MDS code a concept of DRESS (Distributed Replication-based Exact Simple Storage) codes has been introduced in \cite{dress11}. It was an open problem to construct FR codes (and hence dress codes) beyond Steiner systems. Several constructions of Fractional Repetition Codes (and hence dress codes) are known based on bipartite graph \cite{DBLP:journals/corr/abs-1102-3493}, resolvable designs \cite{DBLP:journals/corr/abs-1210-2110} and regular graphs \cite{rr10,Wangwang12}.
%A flexible class of regenerating codes presented in ~\cite{rnv10c}, in which an end user can download the entire file by contacting to any number of nodes as long as the total amount of downloaded data is aleast $B$. 
In this paper, we present the extension of the table based method presented in \cite{rr10} for single node failure. We call it weak fractional repetition codes (and hence weak dress codes). The construction in this paper gives the regenerating code for the values of $n$ and $d$ for which regular graph does not exist. We also present some general constructions of weak dress codes.

%%%%Organization of the paper
The rest of the paper is organized as follows. In Section $2,$ we summarize  popular FR code generation given by El Rouayheb et al.  \cite{rr10} for single node failure. In Section $3,$ we provide a simple construction of FR code for  $\rho = d$ and $n=\theta$.  A construction of weak fractional repetition code from partial regular graph is given in Section $4$.  A ring construction is given in Section $5$ which provides both weak and strong fractional repetition codes. Finally Section $6$ concludes with general remarks.
\section{Previous Work}
 
El Rouayheb et al. \cite{rr10} presented the table based regenerating codes to reduce complexity. This construction consists of containing an outer MDS code with an inner fractional repetition code with repetition degree $ \rho $ that can tolerate up to $ \rho-1 $ nodes failing together. In this section, we summarize the results of \cite{rr10} required for our purpose and also give an example of FR code based on regular graph.
\newtheorem{mydef}{Definition}
\begin{mydef}(Fractional Repetition Code) : A Fractional Repetition (FR) code $\C,$ with repetition degree $\rho$, for an $(n,k,d)$ DSS, is a collection $\C$ of $n$ subsets $U_1,U_2, \ldots ,U_n$ of a set $\Omega = \{1,\ldots,\theta \}$, each having size $d$, i.e, $|U_i|=d$, satisfying the condition that each element of $\Omega$ belongs to exactly $\rho$ sets in the collection.
\label{defFRR}
\end{mydef}
As MDS code exists for any parameters, the challenging part is to get an FR code.  In \cite{rr10}, FR codes of repetition degree $2$ are designed with an uncoded repair based on regular graph which can tolerate at most one node failure. For all feasible values of $n$ and $d$, these codes achieves the capacity $ C_{MBR} $. Either of $n$ and $d$ must be even to get these codes. In fact this is the necessary and sufficient condition for constructing codes from this method \cite{rr10}.
\newtheorem{myprop}{Proposition} 
\begin{myprop} The parameter $\theta$ in Def. \ref{defFRR} of an FR code of degree $\rho$ for an $(n,k,d)$ DSS is given by, 
\begin {equation} \theta \rho = nd. \label {eq:rh} \end {equation} 
\label{prop1}
\end{myprop}
Proposition \ref{prop1} gives a necessary and sufficient condition for existence of FR code \cite{DBLP:journals/corr/abs-1201-3547,rr10}. The following popular results for constructing FR codes are well known from \cite{rr10,dress11}.
\newtheorem{mythm}{Theorem} 
\begin{mythm}
A projective plane of order $m$ gives dress codes with $n=m^2+m+1$ and $\rho = d = m+1$ by taking points as packets and lines as storage nodes.
\end{mythm}

\begin{mythm}
A Steiner system $S(2,\alpha, v)$ gives dress codes with parameters $\rho =\frac{(v-1)}{(\alpha-1)}$ and $n= \frac{\rho v}{\alpha}$ using lines as storage nodes and points as packets.
\end{mythm}

In order to understand our work we now describe briefly the construction of FR codes based on regular  graph \cite{rr10}.

\begin{mydef} (Regular Graph) : A regular graph is a graph where each vertex has the same number of neighbors; i.e. every vertex has the same degree. A $d$-regular graph $R_{n,d}$ of $n$ vertices is the regular graph where all vertices have the same degree $d$. Thus graph $R_{n,d}$ has  $ \frac {nd}{2} $ vertices.
\end{mydef}
\newtheorem{myconst}{Construction} 
\begin{myconst} 
To construct a FR code with repetition degree $\rho = 2$ and with $nd$ even,  Algorithm \ref{alg1}  is well known  \cite{rr10}.
\begin{algorithm}
\caption{FR code construction with $\rho =2$ for $\{n,k,d\}$ DSS, where $nd$ is even}
\begin{enumerate} 
\item Generate a d-regular graph $R_{n,d}$ on n vertices $U_1 , U_2 , \ldots , U_n$. 
\item Index the edges of $R_{n,d}$ from 1 to $ \frac {nd}{2}  $. 
\item Store on node $U_i $ in the DSS the packets indexed by the edges that are adjacent to vertex $U_i$ in the graph.
\end{enumerate}
\label{alg1}
\end{algorithm}
\label{const1}
\end{myconst}
%It is known that the FR codes with repetition degree $\rho  = 2$ obtained by this construction are universally good codes  \cite{rr10}.
\begin{figure}[htb]
\begin{center}
\includegraphics[scale=0.1999999999]{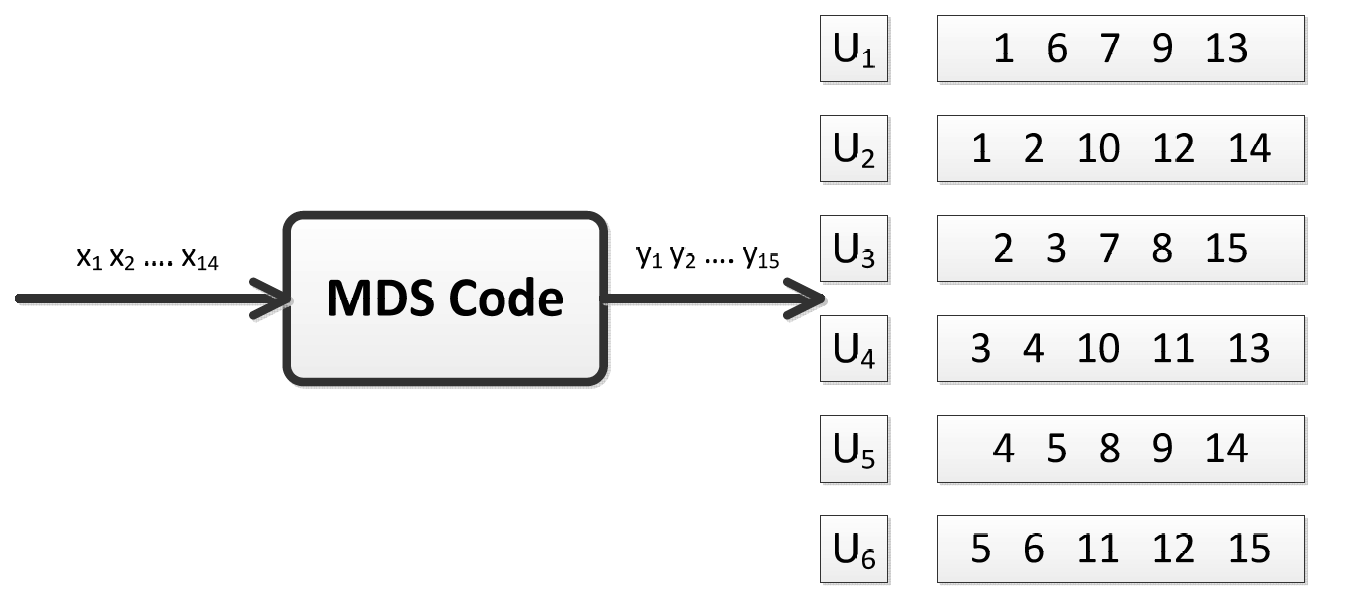}\\
\end{center}
\caption{An exact regenerating code for a $(6,4,5)$ DSS. The code is formed by a $(15, 14)$ parity check MDS code followed by special repetition code. This  inner repetition code is known as a $ Fractional ~Repetition ~Code $ of repetition degree $\rho = 2$.}
\label{f2}
\end{figure}
\begin{figure}[htb]
\begin{center}
\includegraphics[scale=0.40]{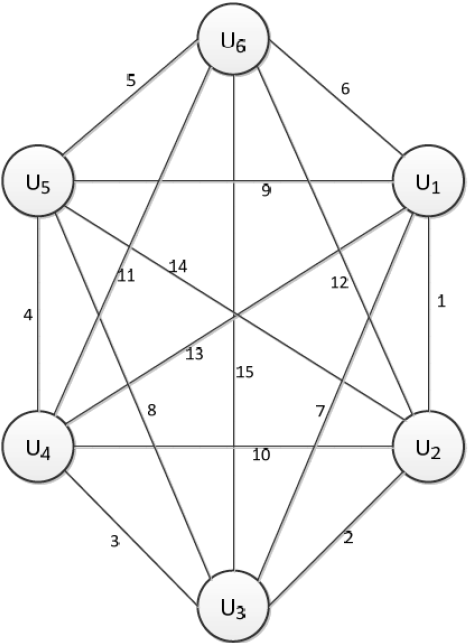}\\
\end{center}
\caption{Regular graph for $n = 6, d = 5$. The labeling of the edges from $1$ to $15$ give the FR code with $\rho = 2$ for DSS (6,4,5) depicted in Figure \ref{f2}.}
\label{f1}
\end{figure}
\newtheorem{myexm}{Example}
\begin {myexm} 
Consider a DSS $( n = 6, k = 4, d = 5 )$ whose storage capacity is $14$ according to (\ref{eq:5}). Let $\F_q$ be the finite field of size $q$. File $\x =(x_1, x_2,\ldots, x_{14}) ~\in ~ \F^{14}_q$ has $14$ packets, which we want to store on the system.  Figure \ref{f2} shows a regenerating code which can save the file of $14$ packets. Regenerating code in Figure \ref{f2} is made of two components- an outer MDS code $(15,14)$ and a repetition code based on a regular graph shown in Figure \ref{f1}. When we give file $\x$ as a input to the outer MDS code, we will get $ \y = (y_1, y_2,\ldots,y_{15})$ as output, where $y_1,y_2,\ldots y_{14}$ are same as $x_1,x_2,\dots,x_{14}$ and $y_{15}$ is the parity packet. The coded data packets will be saved on $6$ different nodes available in the system, according to inner repetition codes of Figure \ref{f2}, i.e. node $U_1$ stores $(y_1,y_6,y_7,y_9,y_{13})$, $U_2$ stores $(y_1,y_2,y_{10},y_{12},y_{14})$, $U_3$ stores $(y_2,y_3,y_7,y_8,y_{15})$, $U_4$ stores $(y_3,y_4,y_{10},y_{11},y_{13})$, $U_5$ stores $(y_4,y_5,y_8,y_9,y_{14})$ and $U_6$ stores $(y_5,y_6,y_{11},y_{12},y_{15})$.  The whole file can be retrieved by recovering just $14$ distinct packets out of these $15$, due to the property of outer MDS code. As $14$ distinct packets can be recovered by contacting to any $4$ nodes out of these $6$ nodes so whole file can be retrieved by contacting to any $4$ nodes. In case of failure of a single node, its data can be recovered by contacting exactly $5$ nodes and downloading $1$ packets from each of these nodes. For instance, when node $U_1$ fails, it contacts node $U_2,U_3,U_4,U_5 ~\mbox{and} ~ U_6$ to get the packets $y_2,y_7,y_{13},y_9 ~ \mbox{and} ~ y_6$ respectively. This example is similar to the example in \cite{rr10}.
\end {myexm}
%\begin{mydef}
%(FR Code Rate) : Rate of an FR Code $R_c(k)$ is the maximum file size, i.e., the maximum number of distinct packets, that the FR code is guaranteed to deliver to any user contacting $k$ nodes. 
%\label {FRCR}
%\end{mydef}
%For the universally good FR Codes 
%\begin {equation}
 %R_c (k) \geq C_{MBR}(n,k,d). 
%\label{eq:6}
%\end {equation}
%\begin{mydef}
%(Fractional Repetition Capacity) : FR Capacity $C_{FR}(k,\rho)$ of a DSS $\{n,k,d\}$ is defined as   
%\begin {equation}
 %C_{FR} (k,\rho) = \mbox {max } R_c (k), 
%\label{eq:7}  
%\end {equation}
%where C is an FR code for DSS $\{n,k,d\}$ with repetition degree $\rho$.
 %\end{mydef}
%From (\ref{eq:6}) and  (\ref{eq:7}), $C_{MBR}$ would be the lower bound on FR capacity. Two Upper Bounds (\ref{eq:8}) and (\ref{eq:10}) on $C_{FR}(n,k,d) $, where regular graph is possible ($nd$ is even) are given in \cite{rr10}. 
%\begin {equation}
 %C_{FR} (k,\rho) \leq \left[\frac{nd}{\rho} \left (1- \frac{{n - \rho \choose k}}{{n \choose k}} \right) \right]. 
 %\label{eq:8}
%\end {equation}
%Following Bound in (\ref{eq:10}) is defined using recursive function which is tighter than the (\ref{eq:8}) . 
%\begin {equation}
%g(1) = d,
%\label{eq:9}
%\end {equation}
%\begin {equation}
%g(k+1) = g(k) + d - \left[ \frac{ \rho g(k) - kd} {n-k} \right], 
%\label{eq:10}
%\end {equation} 
%where FR capacity is upper bounded by the function g(k), i.e. $C_{FR}(k,\rho) \leq g(k) $. For proofs, refer to \cite{rr10}. 

%%%%%%%%%%%%%%%%%%%%%%%%%%%%%%%%%%%%%%%%%%%%
\section {Modular construction of fractional repetition codes for $\rho = d$ and $n=\theta$}
In \cite{rr10}, El Rouayheb et al. have shown the construction of FR codes for repetition degree $3$ or more based on Stainer system. In this section we give a different approach to get the fractional repetition code for $\rho = d$ and  $n=\theta$. 
%While writing the paper recently we became aware of ~\ref{ToniJan12} where a similar technique is described for the existence of fractional repetition code.
%\newtheorem{myprop}{Proposition} 
\begin{myprop}\label{prop2}
Let $n$ and $t (> 1)$ be a positive integers. For $0 \leq j \leq n-1$, let $C^n_j = \{ t^{i-1}+j \pmod {n+1} | 1 \leq i \leq \rho \}$, where $\rho$ is a positive integer such that $t^{\rho-1} < n$. Then $| C^n_j | = \rho$ and $\C = \{ C^n_j | 0 \leq j \leq n-1 \}$ is a FRC over $\Omega = \{1,2, \ldots, n \}$ with $\rho = d$ and $n=\theta$.
\end{myprop}
\begin{proof} By the construction of $C^n_j$ it is clear that each such set will have size $\rho = d$ and since $\theta=n$, the necessary and sufficient condition for FRC (equation (\ref{eq:rh})) is satisfied.
\end{proof}
A fractional repetition code based on Proposition \ref{prop2} for $n = 8, t=2$ and $\rho = 3$ is given in Figure \ref{f5}.
%The motivation for the above construction came from the geometry. For example a FR code can be constructed by taking a triangle as a node and vertices represent a packet in the system. The $3$ vertices in a triangle represent that corresponding $3$ packets are saved on the node represented by a particular triangle.
%\begin {myexm}
%We want to get an FR code for $n = 8$ and $B = 8$. Take a triangle joining 1, 2 and 4 below and rotate it once to get the dotted triangle joining 2, 3 and 5 and continue this process through 6 more rotations. By doing this we get total 8 triangles. each represents a node from 1 to 8. Nodes save the packets - \{\{1,2,4\}, \{2,3,5\}, \{3,4,6\}, \{4,5,7\}, \{5,6,8\}, \{6,7,1\}, \{7,8,2\}, \{8,1,3\}\} on nodes 1 to 8. This is not a stainer system but fulfill all requirements for FR codes. 
%\end{myexm}

%\begin{figure}[htb]
%\begin{center}
%\includegraphics[scale=0.70]{S11.png}\\
%\end{center}
%\caption{Modular construction of fractional repetition code based on triangle}
%\label{f4}
%\end{figure}
\begin{figure}[htb]
\begin{center}
\includegraphics[scale=0.25]{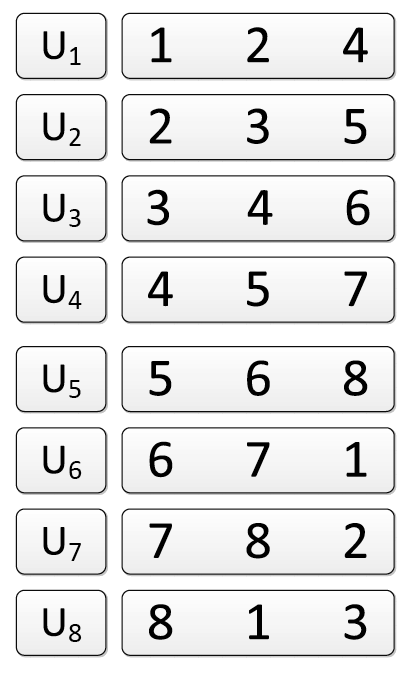}\\
\end{center}
\caption{Fractional repetition code for $n = 8, t=2$ and $\rho = 3$ based on the Proposition \ref{prop2}.}
\label{f5}
\end{figure}

 \hfill
 \hfill 
%%%%%%%%%%%%%%%%%%%%%%%%%%%%%%%%
\section {Construction of Weak Fractional Repetition Codes for $\rho = 2$}\label{WFRro=2}
\begin{mydef}(Weak Fractional Repetition Code) : A Weak Fractional Repetition (WFR) code $\CC$ for a $(n,k,d)$ DSS, with repetition degree $\rho$ for every packet,  is a collection of $n$ subsets $U_1,U_2,\ldots,U_n$ of a set $\Omega = \{1,\ldots,\theta \}$ such that $|U_i| = d_i, ~\mbox{for} ~1 \leq i \leq n$, where $d = \max \{d_i \}, d-d_i = \delta_i$ is the order of weakness of node $U_i$ and $\sum\limits_{i=1}^n {|U_i|} = nd - \delta$, where $\delta =\sum\limits_{i=1}^n \delta_i$ is defined as total order of weakness of $\CC$.
\label{defWFR}
\end{mydef}
\newtheorem{remark}{Remark} 
\begin{remark} WFR code with weakness order $\delta = 0$ gives fractional repetition codes of Def. \ref{defFRR}. \end{remark}
Clearly, for WFR Codes, \begin {equation} \rho \theta = nd- \delta. \end {equation}
FR codes for $\rho = 2$ (WFR code with weakness order 0) has been constructed using regular graph by Ramachandran et al.\cite{rr10}.
Regular graph does not exist for cases where $n$ and $d$ both are odd.  Motivated by this we define a Partial Regular Graph.
\begin{mydef}
(Partial Regular Graph) : A graph $G(n,d)$ with total $n$ vertices such that the degree of some $(n-1)$ vertices is  $d$ and  one vertex has degree $d-1.$ We call these graphs as Partial Regular Graph denoted by $PRG(n,d)$.
\label {defprg}
\end{mydef}
\begin{remark} With the abuse of notations, we will use the same symbols for graph and it's adjacency matrix. \end{remark}
%%%%%%
\subsection{Properties of Partial Regular Graph}
Let a partial regular graph be denoted by $PRG(n,d)$. The rows of its matrix are denoted by $R_i$ and columns are denoted by $C_i$ where $i \in \{1,2, \ldots, n\}$ such that degree of every vertex is $d$ except one vertex $U_k$ for some $k \in \{1,2, \ldots, n\}$  which has degree $d-1$. This adjacency matrix of $PRG(n,d)$ has the following properties. 
\begin{enumerate} 
\item $PRG(n,d)$ is a partial regular graph iff  both $n$ and $d$ are odd.
\item $wt(R_i) = d $ for all $i$ such that $1 \leq i  \leq n$ except for $i=k$ and  $wt(R_k) = d-1$.
\item $wt(C_i) = d $ for all $i$ such that $1 \leq i  \leq n$ except for $i=k$ and $wt(C_k) = d-1$. 
\end{enumerate}
So the total degree of the graph becomes even. In fact it is the necessary condition for the existence of an undirected graph.
%%%%%%%%%%%%%%%%%
\subsection { Construction of the Partial Regular Graph from a Circulant Regular Graph}
\begin{mydef}
(Circulant Graph) : A graph is called as a circulant graph $C_n(d)$ of degree $d$ if it's adjacency matrix is a circulant. i.e. ($a_{i,j}$) = ($a_{j-i}$) for $i,j = 0,\ldots,n-1$ (subscripts of adjacency matrix  are taken modulo $n$), $a_{i,j} \in \F_q$ and weight of each row is $d$.
\end{mydef}
Circulant matrix is characterized by it's $1^{st}$ row. If the first row of the circulant matrix is $(a_0,a_1,\ldots,a_{n-1})  \in \F^n_q$ one can form a polynomial $q(t) = a_0 + a_1t+ a_2 t^2 +\ldots  a_{n-1} t^{n-1}.$ The algebra of circulant matrix has one-to-one correspondance with polynomial algebra in $\F_q [t]/(t^n -1)$.
\begin{remark} Each circulant graph is a regular graph. \end{remark}
%\begin{mydef}
%(Permutation Matrix) : A permutation matrix is a square matrix whose entries are all 0's and 1's, with exactly one $1$ in each row and exactly one $1$ in each column. 
%\end{mydef}
Let  
%\[ 
$
P(\pi) =  \left( \begin{array}{cccc} 
1 & 2 & \ldots & n \\
\pi(1) & \pi(2) & \ldots & \pi(n) \\
 \end{array} \right) 
 $
 %\] 
 be the permutation function on $n$ symbols then the corresponding permutation matrix $P_{n}$ is the $n\times n$ matrix whose entries are all $0$ except that in row $i$, the entry $\pi(i)$ equals $1.$ We know from elementary algebra that to get a permutation matrix $P_{n}$ from a permutation $P(\pi)$, one can apply permutation function $P(\pi)$ on the identity matrix $I_n$.
Algorithm $2$ constructs a Partial Regular Graph $PRG(n,d)$ from a circulant graph $C_n(d-1)$ whose $1^{st}$ row is defined by the polynomial $q(t) = t+ t^2 +\ldots+ t^{(d-1)/2}+ t^{n - (d-1)/2}+\ldots+t^{n-1} $, where weight of $q(t)=d-1$.  
\begin{algorithm}
\caption{Construction of Partial Regular Graph $PRG(n,d)$ from Circulant Matrix $C_n(d-1)$}
\begin{enumerate} 
\item Pick a circulant matrix $C_n(d-1)$ whose first row is defined by $q(t) = t+ t^2 +\ldots+ t^{(d-1)/2}+ t^{n - (d-1)/2}+\ldots+t^{n-1}$.
\item Get a permutation matrix $P_{n-1}$ whose permutation function $P(\pi)$ on $n-1$ symbols is given by
{\footnotesize
\[ \left( \begin{array}{ccccccc} 
1 & 2 & \ldots &\frac{(n-1)}{2} & \frac{(n+1)}{2}& \ldots & n-1\\
\frac{(n+1)}{2} & \frac{(n+1)}{2}+ 1 & \ldots & (n-1) & 1 & \ldots & \frac{(n-1)}{2}\\
 \end{array} \right). \]
 }
\item Expand matrix $P_{n-1}$ to make it $n \times n$ matrix $S_n$ 
by adding $ \bf 0 $ at $n^{th}$ row and  $n^{th}$ column such that 
\begin{math}
S_n =\left( \begin{array}{c|c} P_{n-1} & \bf {0^T} \\\hline \bf {0} & 0 \end{array} \right)
 \end{math}, where $ \bf {0} \rm = [0~ 0 \ldots ~0 ] \in \Z^{n-1}_2$.
\item Add above matrix $S_n$ to $C_n(d-1)$ to get the final matrix $PRG(n,d)$.
\begin {center}
$ PRG(n,d) = S_n + C_n(d-1) \pmod 2$.
\end{center}
\end{enumerate}
\label {alg2}
\end{algorithm}
%%%%%%%%%%%%%
\subsection {Construction of the WFR Code}
WFR code from the partial regular graph can be constructed in the same way as in the previous case of regular graph. Algorithm  \ref{alg2} can be used to construct a a partial regular graph which can be used together with Algorithm \ref{alg3}, to construct a WFR code.  
\begin{algorithm}
\caption{Construction of the Weak FR code from the Partial Regular Graph $PRG(n,d)$ where $n$ and $d$ both are odd.}
\begin{enumerate}
\item Generate the above graph $PRG(n,d)$ on n vertices $ U_1 , U_2 ,\ldots, U_n $. 
\item Index the edges of $G_{n,d}$ from 1 to $ \frac {nd-1}{2}.$
\item Store on node $U_i $ in the DSS the packets indexed by the edges that are adjacent to vertex $U_i$ in the graph.
\end{enumerate}
\label {alg3}
\end{algorithm}
%[width=0.7\textwidth, height=0.5\textheight]
%%%%
\begin{figure}[htb]
\begin{center}
\includegraphics[scale=0.31]{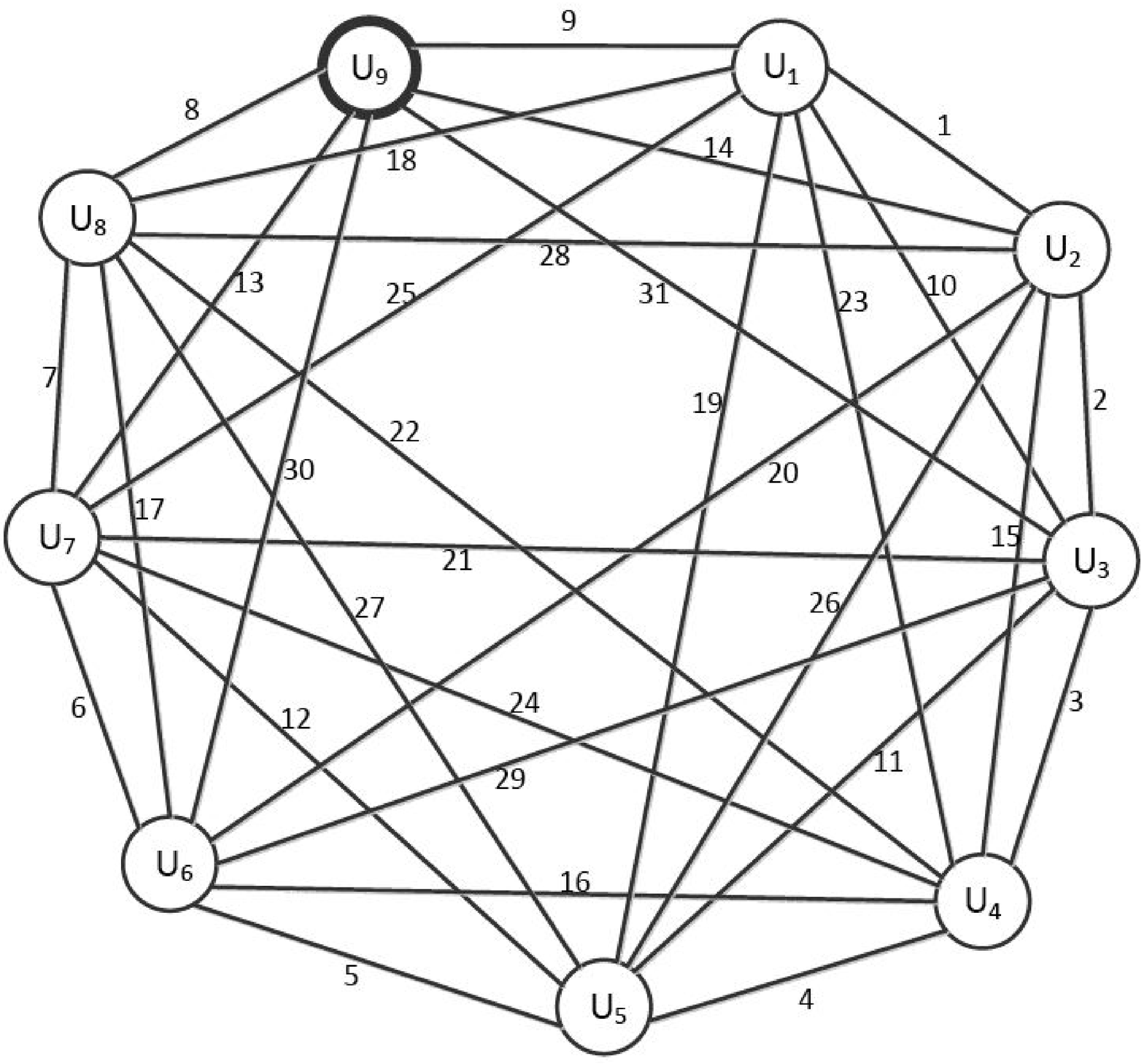}\\
\end{center}
\caption{Fraction Repetition code for $n = 9$ and $d = 7.$ Here All vertices (except $U_9$) are have degree $7$. Each vertex represents a node and edges represents vectors corresponding to the common symbol between the node.}
\label{f3}
\end{figure}
%%%%%%
\begin{myexm}  To generate the Weak Fractional Repetition Code for $n=9$ and $d=7$, note that $q(t) = 0+t+ t^2 + t^3+ 0 t^4+0 t^5+t^6+t^7+t^8$ and hence
\[ 
C_9(6) = \left( \begin{array}{ccccccccc} 
 0  & 1  & 1  & 1 & 0  & 0  & 1  & 1  & 1  \\      
 1  & 0  & 1  & 1 & 1  & 0  & 0  & 1  & 1  \\
 1  & 1  & 0  & 1 & 1  & 1  & 0  & 0  & 1  \\ 
 1  & 1  & 1  & 0 & 1  & 1  & 1  & 0  & 0  \\ 
 0  & 1  & 1  & 1 & 0  & 1  & 1  & 1  & 0  \\ 
 0  & 0  & 1  & 1 & 1  & 0  & 1  & 1  & 1  \\
 1  & 0  & 0  & 1 & 1  & 1  & 0  & 1  & 1  \\ 
 1  & 1  & 0  & 0 & 1  & 1  & 1  & 0  & 1  \\ 
 1  & 1  & 1  & 0 & 0  & 1  & 1  & 1  & 0  \\ \end{array} \right).
 \] 
For the permutation function  
\[P(\pi) = \left( \begin{array}{cccccccc} 
1 & 2 & 3 & 4 & 5 & 6 & 7 & 8 \\
5 & 6 & 7 & 8 & 1 & 2 & 3 & 4 \\
 \end{array} \right)
 \]
the permutation matrix $P_8$ is given by
\[ 
P_8 = \left( \begin{array}{cccccccc} 
 0  & 0  & 0  & 0 & 1  & 0  & 0  & 0  \\ 
 0  & 0  & 0  & 0 & 0  & 1  & 0  & 0  \\ 
 0  & 0  & 0  & 0 & 0  & 0  & 1  & 0  \\ 
 0  & 0  & 0  & 0 & 0  & 0  & 0  & 1  \\ 
 1  & 0  & 0  & 0 & 0  & 0  & 0  & 0  \\ 
 0  & 1  & 0  & 0 & 0  & 0  & 0  & 0  \\
 0  & 0  & 1  & 0 & 0  & 0  & 0  & 0  \\ 
 0  & 0  & 0  & 1 & 0  & 0  & 0  & 0  \\
 \end{array} \right). 
 \]
 Thus by step 3 of Algorithm \ref{alg2}, $S_9$ is given by 
\[ S_9 = \left( \begin{array}
{ccccccccc} 
 0  & 0  & 0  & 0 & 1  & 0  & 0  & 0 &0 \\
 0  & 0  & 0  & 0 & 0  & 1  & 0  & 0 &0 \\
 0  & 0  & 0  & 0 & 0  & 0  & 1  & 0 &0 \\
 0  & 0  & 0  & 0 & 0  & 0  & 0  & 1 &0 \\
 1  & 0  & 0  & 0 & 0  & 0  & 0  & 0 &0 \\
 0  & 1  & 0  & 0 & 0  & 0  & 0  & 0 &0 \\ 
 0  & 0  & 1  & 0 & 0  & 0  & 0  & 0 &0 \\
 0  & 0  & 0  & 1 & 0  & 0  & 0  & 0 &0 \\
 0  & 0  & 0  & 0 & 0  & 0  & 0  & 0 &0 \\ 
\end{array} \right).
\] 
Using step $4$ of Algorithm \ref{alg2} one gets the following adjacency matrix for partial regular graph $PRG(9,7)$ with $n=9$ and $d=7$. 
\[ 
PRG(9,7) = \left( \begin{array}
{ccccccccc} 
 0  & 1  & 1  & 1 & 1  & 0  & 1  & 1  & 1  \\         
 1  & 0  & 1  & 1 & 1  & 1  & 0  & 1  & 1  \\
 1  & 1  & 0  & 1 & 1  & 1  & 1  & 0  & 1  \\  
 1  & 1  & 1  & 0 & 1  & 1  & 1  & 1  & 0  \\  
 1  & 1  & 1  & 1 & 0  & 1  & 1  & 1  & 0  \\  
 0  & 1  & 1  & 1 & 1  & 0  & 1  & 1  & 1  \\ 
 1  & 0  & 1  & 1 & 1  & 1  & 0  & 1  & 1  \\  
 1  & 1  & 0  & 1 & 1  & 1  & 1  & 0  & 1  \\  
 1  & 1  & 1  & 0 & 0  & 1  & 1  & 1  & 0  \\ 
\end{array} \right). 
\] 
Adjacency matrix  $PRG(9,7)$  will give us a graph shown in Figure ~\ref{f3}. It is a partial regular graph which satisfies the properties given in subsection $A.$ This graph will give us the Weak Fractional Repetition Code as shown in the Figure \ref{WFR3192PRG}.

\begin{figure}[htb]
\begin{center}
\includegraphics[scale=0.4]{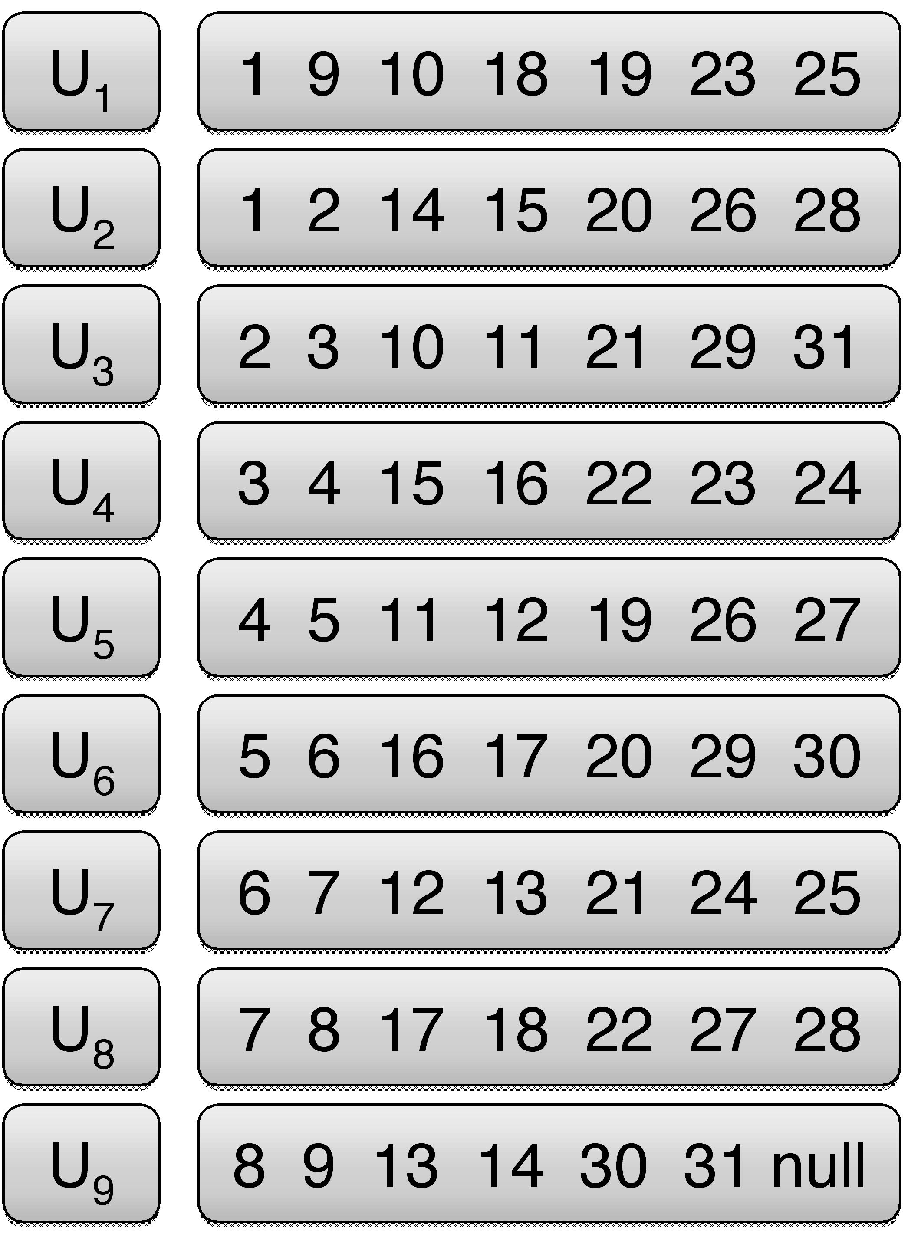}
\end{center}
\caption{Weak Fractional Repetition code based on generalized ring consturction for $\theta=31,\rho=2$ and $n=9$.}
\label{WFR3192PRG}
\end{figure}
%\begin{center}
%\begin{tabular}{|l|c c c c c c c|}
  %\hline
  %\multicolumn{8}{|c|}{Codewords} \\
  %\hline
  %$U_1$ & 1 & 9 & 10 & 18 & 19 & 23 & 25 \\
  %$U_2$ & 1 & 2 & 14 & 15 & 20 & 26 & 28 \\
  %$U_3$ & 2 & 3 & 10 & 11 & 21 & 29 & 31 \\
  %$U_4$ & 3 &4 &15 &16 &22 &23 &24 \\
  %$U_5$ & 4 &5 &11 &12 &19 &26 &27 \\
  %$U_6$ & 5 &6 &16 &17 &20 &29 &30 \\
  %$U_7$ & 6 &7 &12 &13 &21 &24 &25 \\
  %$U_8$ & 7 &8 &17 &18 &22 &27 &28 \\
  %$U_9$ & 8 &9 &13 &14 &30 &31 &null\\  
  %\hline
%\end{tabular}
%\end{center}
In this example, suppose  we want to store a file $\x = (x_1,x_2,x_3, \ldots,x_{30})  \in \F^{30}_q$ on a DSS $(n=9,k=7,d=7)$.  There is no  regular graph which can give us appropriate fractional repetition codes. However Figure \ref{f3} shows a partial regular graph which gives us weak fractional repetition codes to save the file of $30$ packets on DSS $(n=9,k=7,d=7)$. Here regenerating code is made of two components - an outer MDS code $(31,30)$ and a weak fractional repetition code based on partial regular graph shown in Figure \ref{f3}. By taking file $\x$ as a input to the outer MDS code we will get $\y= (y_1,y_2, \ldots ,y_{31})$ as output where $y_1=x_1, y_2=x_2, \ldots ,y_{30}=x_{30}$ and $y_{31}$ is the parity packet. The coded packets will be saved on $9$ different nodes in the system according to partial regular graph shown in Figure \ref {f3}. The whole file can be retrieved by recovering just $30$ distinct packets out of these $31$, due to the property of outer MDS code. As $30$ distinct packets can be recovered by contacting to any $7$ nodes out of these $9$ nodes so whole file can be retrieved by contacting to any $7$ nodes. In case of failure of a single node, its data can be recovered by contacting exactly $7$ nodes and downloading $1$ packets from each of these nodes, except in case of node $U_9$ which has to contact only $6$ nodes in case of failure. For instance, when node $U_1$ fails, it contacts node $U_2,U_3,U_4,U_5,U_7,U_8,U_9$ to get the packets $y_1,y_{10},y_{23},y_{19},y_{25},y_{18},y_9$ respectively. Node $U_9$ saves only $6$ data packets. So we put a default null packet as the $7^{th}$ packet. To recover this node on failure, user needs to contact only $6$ other nodes.
\end{myexm} 
%%%%%%%%%%%%%%%%%%%%%%%%%%%%%%%
\section{Construction of Generalized Ring Code with $\rho=2$}
%%%%%%%%%%%%%%%%%%%%%%%%%%%%%%%
The construction of Section \ref{WFRro=2} can be generalized to give rise both strong and weak fractional repetition codes. Suppose we have $\theta$ packets from set $\Omega = \{1, 2,\ldots,\theta \}$ to be stored on $n$ nodes $U_1,U_2, \ldots ,U_n$ such that $\rho=2$. We first place $n$ nodes on a circle (See Figure \ref{constructionring} for $n=9$ and $\theta = 31$ packets) and then starting from first node start placing the packets between successive nodes until all packets are exhausted. For  $n=9$ and $\theta = 31$ this gives the WFR code as shown in Figure \ref{WFRC31.9}. Here every packet has $\rho=2$.  Since $\theta = q n + r, 0 \leq r \leq n-1$ this simple construction gives FR code for $r=0$ and WFR code for $r > 0$. This simple construction can be modified in different ways to obtained a general class of generalized ring codes. This will be included in the extended version of the paper.
I%f we go through the tables given in Appendix A which shows the maximum file size that can be obtained by product matrix method for regenerating code construction. For $n = 9, k=7, d=7$, the maximum file size that can be retrieved by product matrix construction is $28$. but by table based method we are getting B = 30 for the same. 
%%%%%%%%%%%%%%%%%%%%%
%%%%%%%%%
\begin{figure}[htb]
\begin{center}
\includegraphics[scale=0.37]{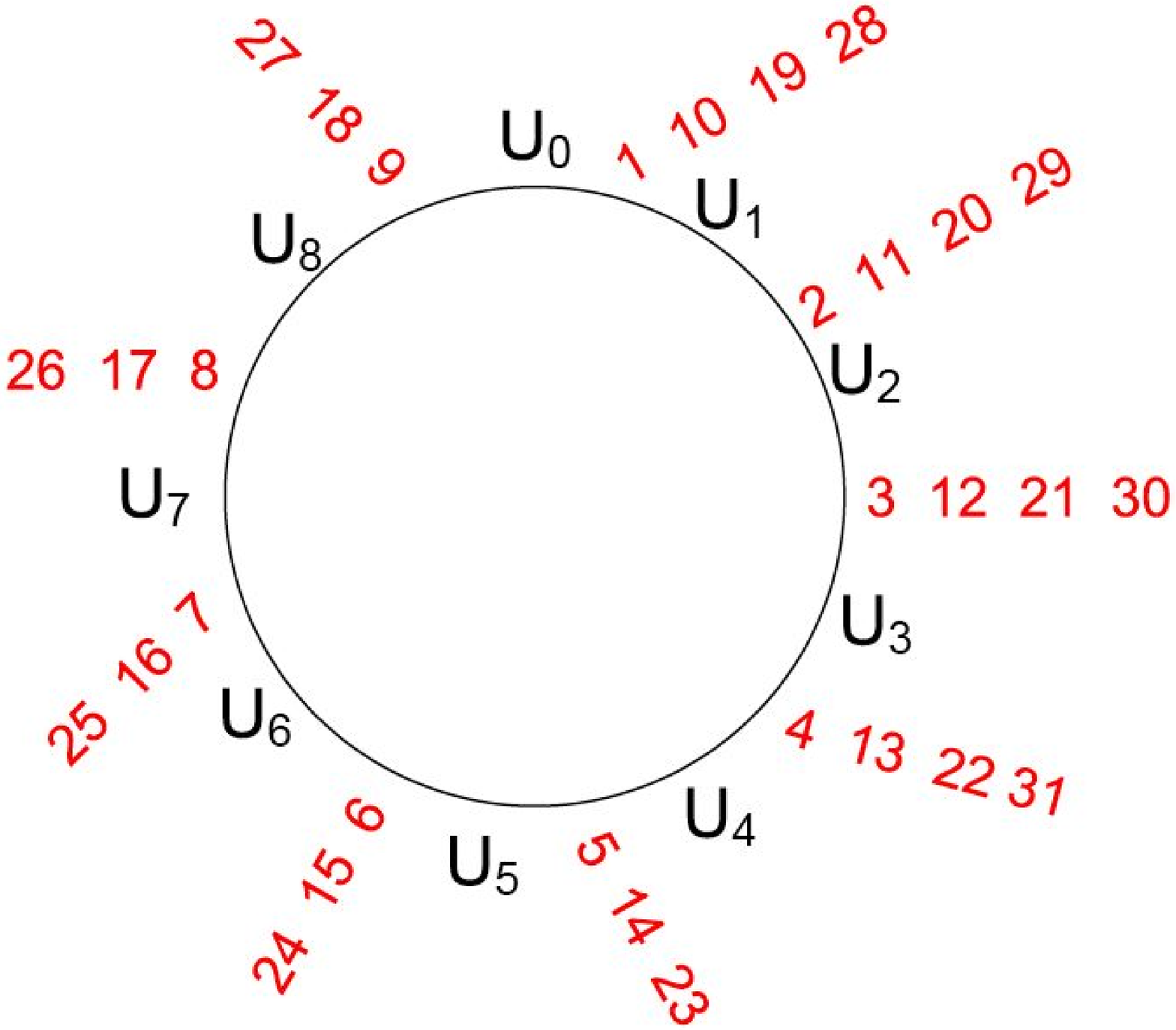}
\end{center}
\caption{Generalized ring construction  for $\theta=31,\rho=2$ and $n=9$.}
\label{constructionring}
\end{figure}

\begin{figure}[htb]
\begin{center}
\includegraphics[scale=0.41]{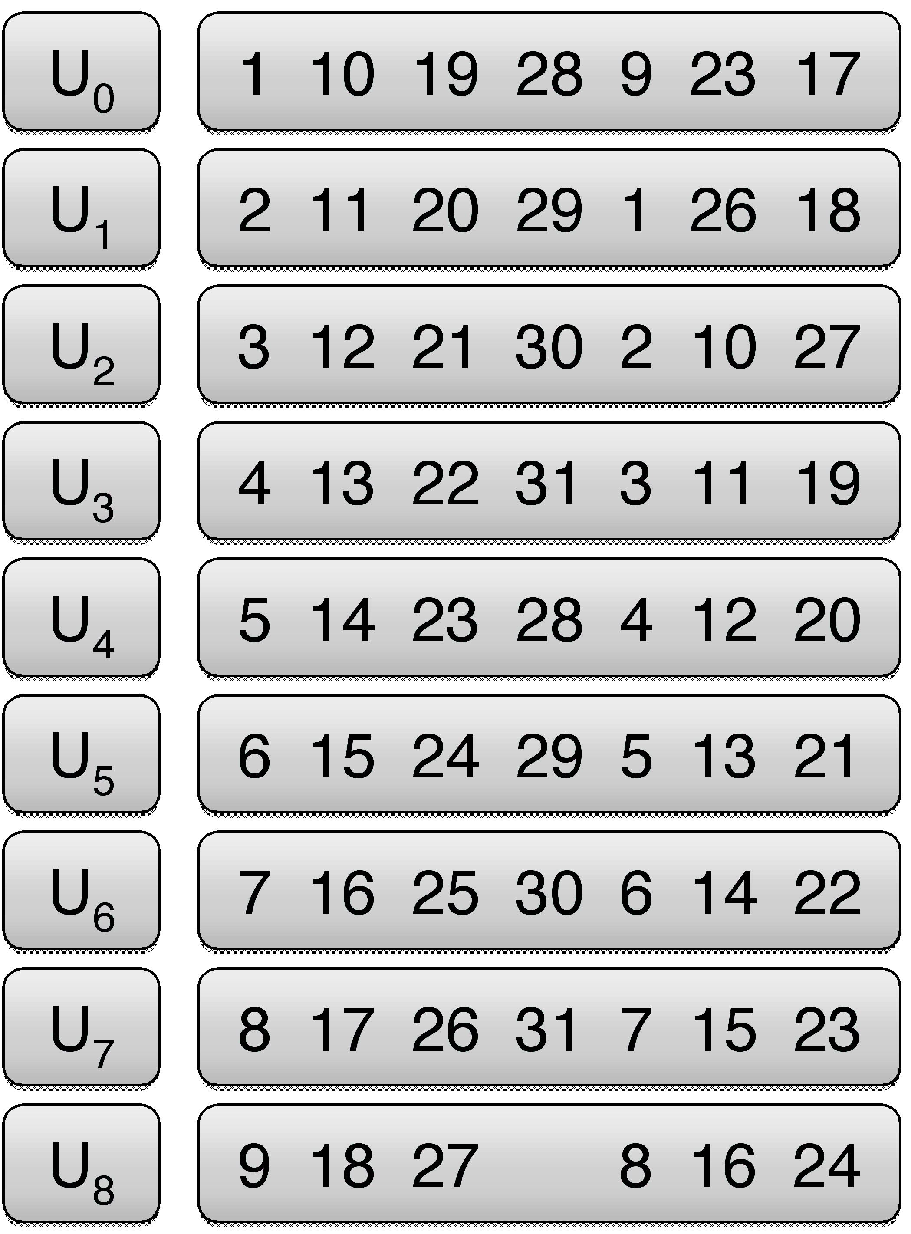}
\end{center}
\caption{Weak Fractional Repetition code based on generalized ring consturction for $\theta=31,\rho=2$ and $n=9$.}
\label{WFRC31.9}
\end{figure}
%%%%%%%%%%
\section{Conclusion}
Motivated by the emergence of fractional repetition codes which were based on regular graphs, we presented weak fractional repetition codes which are based on partial regular graph. WFR code gives regenerating codes for system parameters for which regular graphs does not exist.  In this process, we constructed WFR codes of repetition degree $2$.  As FR code exists where $nd = \mbox{even} $ and we presented WFR codes for $nd = \mbox{odd}$, by combining these two, we get regenerating codes for single node failure for all feasible values of system parameters for any DSS. However WFR codes of higher repetition degree can be used for multiple node failure i.e. $\rho > 2$ as it gives more flexibility to a DSS. In future, weak fractional repetition codes could be constructed by Steiner system but distribution of default null packets remains an open problem. We also presented a simple modular construction of fractional repetition codes for $\rho = d$ and $n=\theta$ and a generalized ring construction that gives both strong and weak fractional repetition codes. It would be an interesting future task to generalize some of these constructions for $\rho > 2$. WFR codes could be useful for heterogeneous distributed storage systems \cite{DBLP:journals/corr/abs-1211-0415}.

\section*{Acknowledgment}
The authors would like to thank Suman Mitra for useful discussions and Ashish Jain, Jay Bhornia and Srijan Anil for drawing some figures. 

\bibliographystyle{IEEEtran}
\bibliography{cloud}
%\begin{thebibliography}{1}
  %0.5em minus 0.4em\relax Harlow, England: Addison-Wesley, 1999.
%\end{thebibliography}

\end{document}